\documentclass[aps,twocolumn,showpacs,amsmath,amssymb]{revtex4}
\usepackage{amsthm}
\usepackage{latexsym}
\usepackage{amsfonts}
\usepackage{bbm,dsfont}
\usepackage{graphicx}

\newtheorem{lemma}{Lemma}

\newtheorem{theorem}{Theorem}
\newtheorem{corollary}{Corollary}
\newtheorem{proposition}{Proposition}

\newcommand{\C}{\mathbb{C}}

\newcommand{\R}{\mathbb{R}}

\newcommand{\h}[1]{\mathcal{#1}}
\newcommand{\hil}{\mathcal{H}}

\newcommand{\sfq}{\mathsf{Q}}
\newcommand{\p}{\mathsf{p}}
 
\newcommand{\E}{\mathsf{E}}

\newcommand{\tr}[1]{\mathrm{tr}\left[ {#1} \right]} 
\newcommand{\de}{{\rm d}}
\DeclareMathOperator{\rank}{rank}

\begin{document}
\title{Informationally complete sets of Gaussian measurements}

\author{Jukka Kiukas}
\address{School of Mathematical Sciences, University of Nottingham, University Park,
Nottingham, NG7 2RD, UK}
\email{jukka.kiukas@nottingham.ac.uk}

\author{Jussi Schultz}
\address{Dipartimento di Matematica, Politecnico di Milano, Piazza Leonardo da Vinci 32, I-20133 Milano, Italy}
\email{jussi.schultz@gmail.com}

\begin{abstract}
We prove necessary and sufficient conditions for the informational completeness of an arbitrary set of Gaussian observables on continuous variable systems with finite number of degrees of freedom. In particular, we show that an informationally complete set either contains a single informationally complete observable, or includes infinitely many observables. We show that for a single informationally complete observable, the minimal outcome space is the phase space, and the observable can always be obtained from the quantum optical $Q$-function by linear postprocessing and Gaussian convolution, in a suitable symplectic coordinatization of the phase space. In the case of projection valued Gaussian observables, e.g., generalized field quadratures, we show that an informationally complete set of observables is necessarily infinite. Finally, we generalize the treatment to the case where the measurement coupling is given by a general linear bosonic channel, and characterize informational completeness for an arbitrary set of the associated observables. 
\end{abstract}

\pacs{03.65.Wj, 03.65.Ta}
\maketitle

\section{Introduction}
The ability to determine an unknown quantum state produced by some source is central for many applications in quantum information science. The procedure of reconstructing the quantum state, known as quantum state tomography, has therefore been under intense investigations and continues to attract a lot of attention \cite{Paris, Lvovsky2009}. In the continuous variable regime, and in particular its quantum optical realizations, there are two commonly used approaches to quantum tomography. In optical homodyne tomography, the set of rotated quadratures is measured using balanced homodyne detection, thus allowing one to ''scan''  the phase space of the system \cite{Vogel1989, Smithey1993}. The alternative method uses the Husimi $Q$-function which can be measured using a double homodyne detection scheme, and has the advantage that the reconstruction requires the measurement of only a single observable \cite{Leonhardt}.

Both of the above instances fall under the class of Gaussian measurements, i.e., measurements which yield a  Gaussian measurement outcome distribution whenever the system is initially in a Gaussian state \cite{Weedbrook2012}. The purpose of this paper is to present a method 
 for investigating whether or not a given set of such Gaussian observables is informationally complete \cite{Prugovecki1977}, i.e.,  allows the reconstruction of an unknown quantum state of the system from the statistics. We consider an $N$-mode electromagnetic field, whose phase space is therefore $2N$-dimensional. We show that by measuring a Gaussian observable one obtains the values of the Weyl transform of the state on a linear subspace of the phase space. Therefore, for a set of observables the union of these subspaces needs to be "sufficiently large" in order for unique state determination to be possible. In particular, we show that if one does not have access to a single informationally complete Gaussian observable, then one necessarily needs infinitely many observables.
 
 After these general results we focus on two specific instances. Firstly, we investigate single informationally complete Gaussian observables in  more detail. We show that if we restrict to the smallest possible dimension of the outcome space, then the set of informationally complete Gaussian observables is exhausted, up to linear transformations of the measurement outcomes, by Gaussian observables which are covariant with respect to phase space translations. Furthermore, we show that in a suitable symplectic coordinatization of the phase space, any informationally complete Gaussian observable with a minimal outcome space is a postprocessing of the $Q$-function. Secondly, we study commutative Gaussian observables which then include projection valued (also called sharp) Gaussian observables as special cases. We  show that no finite set of such observables is informationally complete. For an arbitrary set of  generalized field quadratures, i.e., sharp Gaussian observables with one-dimensional outcome space, we prove a further characterization for informational completeness.  We also find an interesting connection between the generalized quadratures and informationally complete Gaussian phase space observables.  Finally, we consider a more general scenario where the measurement coupling is represented by a general linear bosonic channel. Also in the general case we obtain a characterization for informational completeness and deal explicitly with general covariant phase space observables.

\section{Preliminaries}
The Hilbert space of an electromagnetic field consisting of $N$ bosonic modes is the $N$-fold tensor product $\hil^{\otimes N} = \bigotimes_{k=1}^N \hil_k$, where each single mode Hilbert space is spanned by the number states $\{\vert n\rangle \mid n =0,1,2,\ldots \}$. The creation and annihilation operators related to the $j^{\rm th}$ mode are denoted by $a^*_j$ and $a_j$. In the coordinate representation $\hil_k \simeq L^2  (\R)$ where the number states are represented by the Hermite functions. The Hilbert space of the entire field is then $\hil^{\otimes N} \simeq L^2 (\R^N)$.

The states of the field are represented by positive trace one operators acting on $\hil^{\otimes N}$, and the observables are represented by positive operator valued measures (POVMs) defined on a $\sigma$-algebra of  subsets of some measurement outcome set. In this paper, we will only consider observables taking values in $\R^M$. Each observable is thus  represented by a map $\E:\h B(\R^M)\to\h L(\hil^{\otimes N})$, where $\h B(\R^M)$ is the Borel $\sigma$-algebra of $\R^M$ and $\h L(\hil^{\otimes N})$ denotes the set of bounded operators on $\hil^{\otimes N}$, and which satisfies positivity $\E(X)\geq 0$, normalization $\E(\R^M) =I$ and $\sigma$-additivity $\E(\cup_j X_j)=\sum_j \E(X_j)$ for any sequence of pairwise disjoint sets where the series converges in the weak operator topology. When a measurement of $\E$ is performed on the system initially  prepared in a state $\rho$, the measurement outcomes are distributed according to the probability measure $\p^\E_\rho (X) = \tr{\rho\E(X)}$.

The phase space of the $N$-mode EM-field is $\R^{2N}$ and we use the notation ${\bf x} =  (q_1,p_1,\ldots,q_N,p_N)^T$ for the canonical coordinates. The phase space translations are represented in $\hil^{\otimes N}$ by the Weyl operators $$W({\bf x}) = e^{-i{\bf x}^T{\bf \Omega} {\bf R} }$$ where the matrix
$$
{\bf \Omega} = \bigoplus_{j=1}^N \omega_j ,\qquad \omega_j = \left(\begin{array}{cc}   0 & 1 \\ -1 & 0 \end{array}\right)
$$
determines the symplectic form $({\bf x}, {\bf y})\mapsto {\bf x}^T{\bf\Omega} {\bf y}$ and ${\bf R} = (Q_1,P_1,\ldots,Q_N,P_N)^T$, with $Q_j=\frac{1}{\sqrt{2}} (a^*_j+a_j)$ and $P_j=  \frac{i}{\sqrt{2}} (a^*_j-a_j)$ being the canonical quadrature (i.e., position and momentum) operators acting on the $j^{\rm th}$ mode. Whenever appropriate, we will emphasize the number of modes in question with a subscript so that, e.g., in the above case we would have ${\bf \Omega}_N$. Due to the Stone-von Neumann theorem, the Weyl operators are determined by the relation 
\begin{equation}\label{eqn:weyl_relation}
W({\bf x})W({\bf y}) = e^{-\frac{i}{2}{\bf x}^T {\bf \Omega} {\bf y}}W({\bf x+y})
\end{equation}
up to unitary equivalence. This unitary freedom corresponds to the choice of the canonical coordinates: let ${\rm Sp} (2N)$ denote the group of symplectic transformations of $\R^{2N}$, i.e., real invertible $2N\times 2N$-matrices ${\bf S}$ satisfying ${\bf S}^T {\bf \Omega S} = {\bf \Omega}$. If we change the coordinates as ${\bf x}\mapsto {\bf \tilde x} ={\bf S}^{-1}{\bf x}$,
then the original Weyl operators in the new coordinates are given by $W_{\bf S}({\bf \tilde x})=W({\bf x})=W({\bf S\tilde x})$. Since $\bf S$ is symplectic, the operators $W_{\bf S}({\bf \tilde x})$ indeed also satisfy the Weyl relation \eqref{eqn:weyl_relation}, and there exists a unitary operator $U({\bf S})$ such that 
\begin{equation}\label{eqn:weyl_symplectic}
W({\bf S\tilde x}) = U({\bf S}) W({\bf \tilde x}) U({\bf S})^*
\end{equation}
for all ${\bf \tilde x}\in\R^{2N}$. Summarizing, a change of canonical coordinates is effected by the transformation
\begin{align}\label{eqn:coordinate_transformation}
{\bf x}&\mapsto {\bf S}^{-1}{\bf x}, & A&\mapsto U({\bf S}) AU({\bf S})^*,
\end{align}
where ${\bf S}\in {\rm Sp} (2N)$.

A set of observables $\{\E_j \mid j\in\h I\}$ is called informationally complete \cite{Prugovecki1977} if an arbitrary unknown state $\rho$ is uniquely determined by the collective measurement outcome statistics $\{ \p_\rho^{\E_j}\mid j\in\h I\} $, or equivalently, by the corresponding Fourier transforms 
$$
\widehat{\p}_\rho^{\E_j} ({\bf y}) = \int e^{i{\bf y}^T{\bf x}} \, \p_\rho^{\E_j} (\de{\bf x}).
$$
For our purposes, the crucial fact is that any state $\rho$ is uniquely determined by its Weyl transform $$\widehat{\rho}\,({\bf x}) = \tr{\rho W({\bf x})},$$ which in our case turns out to be directly related to the Fourier transforms of the outcome distributions. The Weyl transform is a continuous function, so that $\rho$ is already determined if we know the values of $\widehat{\rho}$ on a dense subset of $\R^{2N}$. The following lemma will then turn out to play a central role. It expresses precisely the intuitive idea that in order to uniquely determine the state, one has to "scan" the phase space sufficiently well. Here $\h T(\hil^{\otimes N})$ refers to the set of trace class operators on $\hil^{\otimes N}$.

\begin{lemma}\label{lemma:open_set}
Let $U\subset \R^{2N}$ be a nonempty open set. Then there exists a nonzero $S\in\h T(\hil^{\otimes N})$ such that $\widehat{S}({\bf x}) = 0$ for all ${\bf x}\notin U$.
\end{lemma}
\begin{proof}
Let $S_0\in\h T(\hil^{\otimes N})$ be such that $\widehat{S_0}({\bf x}) \neq 0$ for all ${\bf x}\in\R^{2N}$ (e.g., the vacuum state). For any $f\in L^1 (\R^{2N})$ define the operator-valued convolution $f*S_0\in\h T(\hil^{\otimes N})$ \cite{Werner1984} via
$$
f*S_0  = \int f({\bf x}) W({\bf x}) S_0 W({\bf x})^* \, \de{\bf x}
$$
so that $\widehat{f*S_0}({\bf x}) = \widehat{f}({\bf \Omega}{\bf x})\,  \widehat{S_0}({\bf x}) = 0$ if and only if $\widehat{f}({\bf \Omega}{\bf x})=0$. We can now choose $\widehat{f}$ to be a nonzero compactly supported $C^\infty$-function such that $\widehat{f}({\bf \Omega}{\bf x})=0$ for all ${\bf x}\notin U$. We then obtain $f$ via the inverse Fourier transform, and our choice of $\widehat{f}$ guarantees that $f\in L^1 (\R^{2N})$. Thus, $S=f*S_0$ is of the required type. 
\end{proof}

\section{Gaussian observables and their postprocessings}

We will next recall the relevant concepts of Gaussian states, channels, and observables, as well as the concept of postprocessing of a measurement. For more details on continuous variable quantum information, we refer the readers to  \cite{Weedbrook2012, Braunstein2005} and   \cite[Ch. V]{Holevo}.

A state $\rho$ is called a Gaussian state, if the Weyl transform $\widehat{\rho}$ is of  a Gaussian form. More specifically, the Weyl transform of a Gaussian state is given by
\begin{equation*}
\widehat{\rho} \, ({\bf x})  = e^{-\frac{1}{4}{\bf x}^T ({\bf \Omega}^T{\bf V  \Omega} ){\bf x} - i ({\bf \Omega m})^T{\bf x}}
\end{equation*}
where ${\bf m}$ is the displacement vector whose components are given by the first moments $m_j = \tr{\rho R_j}$ and ${\bf V}$ is the covariance matrix whose elements are $V_{ij} =   \tr{\rho \{ R_i-m_i,R_j-m_j \}}$ where $\{\cdot,\cdot\}$ denotes the anticommutator. The covariance matrix is a real symmetric matrix which satisfies the uncertainty relation \cite{Simon1994}
\begin{equation*}
{\bf V} + i{\bf \Omega} \geq {\bf 0}.
\end{equation*}

A Gaussian channel is a completely positive trace preserving map $\Phi:\h T(\h H^{\otimes N}) \to \h T(\h H^{\otimes M}) $ which maps Gaussian states into Gaussian states. Written in terms of the Weyl transforms of the input and output states, any Gaussian channel can be characterized by
\begin{equation}\label{eqn:gaussian_weyl_transform}
\widehat{\Phi(\rho)} ({\bf x}) = \widehat{\rho}\,  ({\bf Ax}) e^{-\frac{1}{4}{\bf x}^T {\bf B x} - i {\bf v}^T{\bf x}}
\end{equation}
where ${\bf A}$ is a $2N\times 2M$-matrix and ${\bf B}$ is a $2M\times 2M$-matrix which must satisfy the  complete positivity condition 
\begin{equation}\label{eqn:gaussian_complete_positivity}
{\bf B} + i{\bf \Omega}_M - i{\bf A}^T {\bf \Omega}_N {\bf A} \geq {\bf 0}
\end{equation}
and ${\bf v}\in\R^{2M}$ is a fixed vector \cite{Demoen1979, Holevo2001, Holevo}. Equivalently, we may use the dual channel $\Phi^*:\h L(\h \hil^{\otimes M})\to \h L(\hil^{\otimes N})$, whose Gaussian character is captured in its action on the Weyl operators 
\begin{equation*}
\Phi^*(W ({\bf x})) = W ({\bf Ax}) e^{-\frac{1}{4}{\bf x}^T {\bf B x} - i {\bf v}^T{\bf x}}.
\end{equation*}

We say that an observable $\E:\h B(\R^M)\to\h L(\hil^{\otimes N})$ is Gaussian, if the measurement outcome distribution is Gaussian whenever the system is initially in a Gaussian state. This can be conveniently expressed in terms of the characteristic function of the POVM 
$$
\widehat{\E}({\bf p}) = \int e^{i{\bf p }^T{\bf x}} \, \E(\de{\bf x}),
$$
in which case $\E$ is Gaussian whenever 
\begin{equation}\label{eqn:gaussian_characteristic_function}
\widehat{\E}({\bf p})  =  W({\bf A}_0 {\bf p}) e^{-\frac{1}{4} {\bf p}^T {\bf B}_0 {\bf p} - i {\bf v}_0^T {\bf p}}
\end{equation}
where ${\bf A}_0$ is a $2N\times M$-matrix and ${\bf B}_0$ is  a $M\times M$-matrix which satisfy the positive definiteness condition 
\begin{equation}\label{eqn:gaussian_povm_condition}
{\bf B}_0 - i {\bf A}^T_0 {\bf \Omega}_N {\bf A}_0 \geq {\bf 0}
\end{equation}
and ${\bf v}_0\in\R^M$.  Note that our definition differs from the one used, e.g., in \cite{Holevo} where the outcome space of a Gaussian observable is assumed to be a symplectic space, i.e., even dimensional. The physical motivation for this definition will become apparent in the next section where measurement schemes realizing Gaussian observables are discussed.

A postprocessing of a measurement is a fixed transformation performed on the measurement outcome distribution. This is relevant for informational completeness, because information can only get lost in such a process; if a postprocessing is informationally complete, then so is the original measurement. Here we consider two different kinds of postprocessings: linear postprocessings and smearings. To this end, let  $\E:\h B(\R^M)\to\h L(\hil^{\otimes N})$ be a Gaussian observable parametrized by $({\bf A}_0,{\bf B}_0,{\bf v}_0)$.

Given a $M'\times M$ matrix ${\bf P}$  we define the observable $\E_{\bf P}:\h B(\R^{M'})\to\h L(\hil^{\otimes N})$ via
$$\E_{\bf P}(X)=\E({\bf P}^{-1} (X)),$$
where ${\bf P}^{-1}(X)=\{{\bf y}\in \R^M\mid {\bf P}{\bf y} \in X \}$. Then $\E_{\bf P}$ is clearly Gaussian, and it is called a linear postprocessing of $\E$. If ${\bf P}$ is invertible, $\E_{{\bf P}}$ is called a bijective linear postprocessing of $\E$. A straightforward computation shows that the triple of parameters characterizing $\E_{\bf P}$ is $({\bf A}_0{\bf P}^T,{\bf P}{\bf B}_0{\bf P}^T,{\bf P}{\bf v}_0 )$. It is clear from the definition that a bijective linear postprocessing of $\E$ is informationally complete if and only if $\E$ is. 

For any probability measure $\mu:\h B(\R^M)\to[0,1]$, we define the observable $\mu*\E:\h B(\R^M)\to \h L(\hil^{\otimes N})$ via 
$$(\mu*\E)(X) = \int \mu(X-{\bf x})\, \E(\de{\bf x}),$$
and say that $\mu*\E$ is a smearing of $\E$. Note that for any state $\rho$ the corresponding probability distribution is just a convolution of the original one with the measure $\mu$, i.e., ${\bf p}^{\mu*\E}_\rho = \mu * {\bf p}^\E_\rho$. If $\mu$ is Gaussian, i.e, $\widehat{\mu}({\bf p}) = e^{-\frac{1}{4}{\bf p}^T {\bf Cp}- i{\bf d}^T{\bf p}}$ with ${\bf C}\geq {\bf 0}$, then the smeared observable is also Gaussian and, using the fact that $\widehat{\mu*\E}=\widehat{\mu}\, \widehat{\E}$, we find that  the parameters of the smeared observable are $({\bf A}_0, {\bf B}_0+{\bf C}, {\bf v}_0 + {\bf d})$. Since $\widehat{\mu}$ is nonzero everywhere, the smearing $\mu*\E$ is informationally complete if and only if $\E$ is.
 
\section{Unitary dilations of Gaussian observables}\label{sec:dilation}
We will next construct a unitary measurement dilation for an arbitrary Gaussian observable. This is done by first showing that any Gaussian observable can be measured by applying  a Gaussian channel to the field and then performing homodyne detection on the output of the channel. Since unitary dilations of Gaussian channels are known \cite{Caruso2008}, this then allows us to construct the desired measurement dilation.

Suppose first that we have a Gaussian channel $\Phi:\h T(\hil^{\otimes N}) \to \h T(\hil^{\otimes M})$ determined by the parameters $({\bf A}, {\bf B}, {\bf v})$. Let $\sfq:\h B(\R^M)\to\h L(\hil^{\otimes M})$ be the canonical spectral measure, i.e., $\sfq (X)$ corresponds to multiplication by the indicator function of the set $X$, and define the observable $\E: \h B(\R^M)\to\h L(\hil^{\otimes N})$ as $\E (X)  = \Phi^*(\sfq(X))$. The characteristic function of $\E$ is then $ \widehat{\E} ({\bf p})  = \Phi^* \left(e^{i{\bf p}^T{\bf Q}}\right)$ where ${\bf Q} = (Q_1,\ldots,Q_M)^T$. For any ${\bf p}\in\R^M$ we denote  ${\bf p}_0 = (0,p_1,\ldots,0,p_M)^T\in\R^{2M}$ so that $e^{i{\bf p}^T{\bf Q}} = W({\bf p}_0)$, and we obtain
\begin{equation*}
\widehat{\E} ({\bf p}) =  W ({\bf Ap}_0) e^{-\frac{1}{4}{\bf p}_0^T {\bf B p}_0 - i {\bf v}^T{\bf p}_0}
\end{equation*} 
by Eq.~\eqref{eqn:gaussian_weyl_transform}. Finally, we define the $2N\times M$-matrix ${\bf A}_0$, the  $M\times M$-matrix ${\bf B}_0$ and the vector ${\bf v}_0\in\R^M$ by 
\begin{equation}\label{eqn:povm_matrix}
({\bf A}_0)_{ij} = {\bf A}_{i,2j},\quad ({\bf B}_0)_{ij} = {\bf B}_{2i,2j},\quad ({\bf v}_0 )_{i} = {\bf v}_{2i}
\end{equation}
so that the characteristic function can be expressed as
\begin{equation*}
\widehat{\E}  ({\bf p}) =  W  ({\bf A}_0{\bf p}) e^{-\frac{1}{4}{\bf p}^T {\bf B}_0 {\bf p} - i {\bf v}^T_0{\bf p}}.
\end{equation*}
In other words, the observable is Gaussian.

Conversely, suppose that we have a Gaussian POVM $\E$ determined by the parameters $({\bf A}_0, {\bf B}_0, {\bf v}_0)$. We need to find parameters $({\bf A}, {\bf B}, {\bf v})$ of a Gaussian channel such that Eqs.~\eqref{eqn:povm_matrix} hold. To this end, first define  the matrix ${\bf B}'$ by setting  ${\bf B}'_{2i,2j} = ({\bf B}_0)_{ij} $ and ${\bf B}_{ij}'=0$ otherwise, and then set ${\bf B}= {\bf B}' - i{\bf \Omega}_M$ so that  ${\bf B}_{2i,2j}= ({\bf B}_0)_{ij} $. Define also the matrix ${\bf A}$ via ${\bf A}_{i,2j} = ({\bf A}_0)_{ij}$ and ${\bf A}_{ij}=0$ otherwise,  and the vector ${\bf v}$ in a similar manner as ${\bf v}_{2i}=({\bf v_0})_i$ and ${\bf v}_i=0$ otherwise. In order to prove the validity of the complete positivity condition \eqref{eqn:gaussian_complete_positivity}, we note that by our construction this reduces to showing that ${\bf B}' - i{\bf A}^T {\bf \Omega}_N {\bf A} \geq {\bf 0}$ which is an immediate concequence of \eqref{eqn:gaussian_povm_condition} and the definitions of the matrices in question.

In order to reach the desired unitary dilation, we recall that any Gaussian channel can be realized by coupling the $N$-mode system to $L$ auxiliary Gaussian modes, and then applying two types of unitary operators to the total system: symplectic transformations $U({\bf S})$ with ${\bf S}\in {\rm Sp} (2(N+L))$, and displacements $W({\bf d})$ with ${\bf d}\in \R^{2(N+L)}$ \cite{Caruso2008} (see Fig.~\ref{fig:dilation}). The number of auxiliary modes can always be chosen so that $L\leq 2\max\{N,M \}$, though the optimal choice depends on the details of the channel  \cite{Caruso2011}.  The rest of this section is devoted to demonstrating how the parameters of the Gaussian channel, and hence the corresponding observable, are determined by the dilation of the channel.

\begin{figure}
\includegraphics[width=9cm]{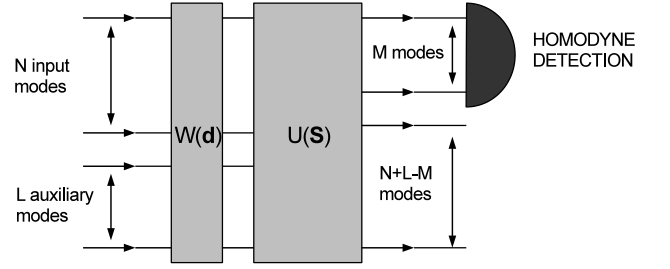}
\caption{Schematic of a measurement of a Gaussian observable. The $N$ input modes are coupled to $L$ auxiliary modes by two unitary couplings: a displacement $W({\bf d})$ and a symplectic unitary $U({\bf S})$. After the coupling, homodyne detection is performed on $M$ output modes, while the other $N+L-M$ modes are discarded.}\label{fig:dilation}
\end{figure}

Now suppose that the auxiliary $L$-mode field is in some Gaussian state $\sigma$ with covariance matrix ${\bf V}$ and displacement vector ${\bf m}$. Since we are interested only in $M\leq N+L$ modes, we must discard the other $K=N+L-M$ modes after the interaction. The resulting channel is then given by 
$$
\Phi(\rho) = \textrm{tr}_{ K}\left[W({\bf d})^* U({\bf S})^* (\rho \otimes \sigma) U({\bf S}) W({\bf d}) \right]
$$
where $\textrm{tr}_{K}[\cdot]$ denotes partial trace over the Hilbert space $\hil^{\otimes K}$ of the discarded modes. Without loss of generality we may assume the ordering of the Hilbert spaces to be fixed as
$$
\hil^{\otimes (N+L)} \simeq \hil^{\otimes N} \otimes \hil^{\otimes L} \simeq \hil^{\otimes M} \otimes \hil^{\otimes K}.
$$
The Weyl transform of the output state is then 
$$
\widehat{\Phi(\rho) } ({\bf x}) = \tr{ W({\bf d})^* U({\bf S})^* (\rho \otimes \sigma) U({\bf S}) W({\bf d}) W({\bf x}) \otimes I}
$$
so  that by denoting ${\bf x}'=(x_1,\ldots,x_{2M},0,\ldots,0)^T\in\R^{2(N+L)}$ we have  
$$
\widehat{\Phi(\rho) } ({\bf x})  = e^{-i{\bf d}^T{\bf \Omega x}'}(\widehat{\rho\otimes \sigma}) ({\bf Sx}').
$$
If we now write ${\bf S}$ in block form as
$$
{\bf S} = \left(\begin{array}{cc} {\bf S}_{11} & {\bf S}_{12} \\ {\bf S}_{21} & {\bf S}_{22} \end{array}\right)
$$
where ${\bf S}_{11}$ is a $2N\times 2M$-matrix and ${\bf S}_{21} $ is a $2L\times 2M$-matrix, and define $\widetilde{{\bf d}} = (d_1,\ldots, d_{2N})^T$ we have 
\begin{eqnarray*}
\widehat{\Phi(\rho) }  ({\bf x})  &=& e^{-i\widetilde{{\bf d}}^T{\bf \Omega x}}\, \widehat{\rho} ({\bf S}_{11}{\bf x})\,  \widehat{\sigma} ({\bf S}_{21}{\bf x})\\
&=&  \widehat{\rho} ({\bf A}{\bf x}) e^{-\frac{1}{4}{\bf x}^T {\bf Bx}-i {\bf v}^T{\bf x}}
\end{eqnarray*}
where $ {\bf A} = {\bf S}_{11}$, ${\bf B} = {\bf S}_{21}^T {\Omega}^T {\bf V\Omega} {\bf S}_{21}$, and ${\bf v} = {\bf S}_{21}^T {\bf \Omega m} - {\bf\Omega}\widetilde{{\bf d}}$. Finally, the parameters $({\bf A}_0, {\bf B}_0, {\bf v}_0)$ of the corresponding  Gaussian observable are obtained from these via Eqs.~\eqref{eqn:povm_matrix}.

\section{Characterization of informational completeness}\label{sec:informational_completeness}
We will next show a general method for determining the informational completeness of a given set of Gaussian observables. To this end, let $\E:\h B(\R^M)\to\h L(\hil^{\otimes N})$ be a Gaussian POVM with parameters $({\bf A}_0, {\bf B}_0 , {\bf v}_0)$. For an arbitrary state $\rho$, the Fourier transform of the probability measure $\p^{\E}_\rho(X) =\tr{\rho \E(X)}$ then reads
\begin{equation}\label{eqn:gaussian_fourier}
\widehat{\p}^{\E}_\rho ({\bf p})  = \tr{\rho \widehat{\E}({\bf p})} =  \widehat{\rho}\,  ({\bf A}_0{\bf p}) e^{-\frac{1}{4}{\bf p}^T {\bf B}_0 {\bf p} - i {\bf v}^T_0{\bf p}},
\end{equation}
and by linearity, \eqref{eqn:gaussian_fourier} also holds with $\rho$ replaced by a general trace class operator $S$ in which case the measure is complex valued. 
Since the Gaussian term on the right-hand-side of  Eq.~\eqref{eqn:gaussian_fourier} is non-zero, we can divide both sides by it. Hence, by measuring $\E$ we are able to  determine the values $\widehat{\rho}({\bf x})$ for all ${\bf x}$ in the set
$$
X_{\E} = \{  {\bf A}_0 {\bf p} \mid  {\bf p}\in\R^M\},
$$  
which is clearly a subspace of $\R^{2N}$. We are now ready to prove the main general result of this paper.

\begin{theorem}\label{prop:IC_set_of_gaussian_povms}
A set of Gaussian observables $\{\E_j :\h B(\R^{M_j})\to\h L(\hil^{\otimes N})\mid j\in\h I\}$ is informationally complete if and only if $\bigcup_{j\in\h I} X_{\E_j}$ is dense in $\R^{2N}$. 
\end{theorem}
\begin{proof}
If $\bigcup_{j\in\h I} X_{\E_j}$ is dense in $\R^{2N}$, then by measuring the observables one can determine the values of $\widehat{\rho}$ on a dense set, so that by the continuity of the Weyl transform, $\rho$ is uniquely determined. Conversely, assume that $\bigcup_{j\in\h I} X_{\E_j}$ is not dense. Then there exists a nonempty open set $U$ in the complement of the closure of $\bigcup_{j\in\h I} X_{\E_j}$, so by Lemma~\ref{lemma:open_set}, there exists a nonzero $S\in\h T(\hil^{\otimes N})$ whose Weyl transform vanishes outside $U$. By \eqref{eqn:gaussian_fourier} and the injectivity of the Fourier transform we then have $\tr{S\E_j (X)} =0$ for each $j$, and all $X\in\h B(\R^{M_j})$. Since the POVM elements are positive, the same holds with $S$ replaced by either of the traceless selfadjoint operators $S+S^*$ and $i(S-S^*)$. At least one of these is nonzero, and hence a constant times the difference of two distinct density operators. By definition, this implies that the set $\{\E_j \mid j\in\h I\}$ is not informationally complete.
\end{proof}

The following two results now follow immediately.

\begin{corollary}\label{cor:IC_finite_gaussian}
If $\h I$ is a finite set, then $\{\E_j \mid j\in\h I\}$ is informationally complete if and only if $\E_{j_0}$ is informationally complete for some $j_0\in\h I$. 
\end{corollary}
\begin{proof} The "if''-part is trivial. Since each $X_{\E_j}$ is closed and $\h I$ is finite, also the union $\bigcup_{j\in\h I} X_{\E_j}$ is closed. Hence, $\{\E_j \mid j\in\h I\}$ is informationally complete only if $\bigcup_{j\in\h I} X_{\E_j} = \R^{2N}$. But this is possible only if  $X_{\E_{j_0}}=\R^{2N}$ for some $j_0\in\h I$, because otherwise each subspace $X_{\E_j}$ would have zero Lebesgue measure and by subadditivity, so would also $\R^{2N}$. An application of Theorem \ref{prop:IC_set_of_gaussian_povms} completes the proof.
\end{proof}

\begin{corollary}\label{cor:IC_gaussian_povm}
A Gaussian observable $\E$ is informationally complete if and only if $\rank {\bf A}_0=2N$.
\end{corollary}
\begin{proof}
Being a subspace, $X_\E$ is closed, so $\E$ is informationally complete if and only if $X_{\E}=\R^{2N}$, i.e., $\rank {\bf A}_0=2N$.
\end{proof}

We now proceed to investigate the informational completeness of a single Gaussian observable more carefully.

\begin{proposition}\label{prop:minimal}
Let $X\subset \R^{2N}$ be a linear subspace and denote $M=\dim X\leq 2N$. Then there exists a Gaussian observable $\E:\h B(\R^M)\to \h L(\hil^{\otimes N}) $ such that $X_\E=X$. In particular, for each $N$ there exist both informationally complete an informationally incomplete Gaussian observables.
\end{proposition}
\begin{proof}
Let $\{ {\bf e}_j \mid j=1,\ldots, M \}$ be an orthonormal basis of $X$ and define ${\bf A}_0= \left( {\bf e}_1 \ldots {\bf e}_M\right)$ so that ${\bf A}_0$ is a $2N\times M$ -matrix.  Since $-i{\bf A}_0^T\Omega_N{\bf A}_0$ is Hermitian, there exists a matrix ${\bf B}_0$, such that \eqref{eqn:gaussian_povm_condition} is satisfied. Then for any ${\bf v}_0\in\R^M$ the parameters $({\bf A}_0, {\bf B}_0, {\bf v}_0)$ define a Gaussian observable $\E$ which satisfies $X_\E =X$. The last claim then follows by choosing $M=2N$ and, e.g., $M=1$.
\end{proof}

Prop.~\ref{prop:minimal} shows that if the dimension of $X_\E$ is strictly smaller then the dimension of the measurement outcome space of $\E$, then from the state reconstruction point of view there is some redundancy in the measurement setup. Indeed, the same information about the state can be obtained with a smaller outcome space and hence, with fewer homodyne measurements. In particular, for an $N$-mode field it is necessary but also sufficient to couple the field to $N$ auxiliary modes, thus giving the outcome space dimension $2N$. In this case, we say that an observable satisfying the rank condition of Cor.~\ref{cor:IC_gaussian_povm} is a minimal informationally complete Gaussian observable. Clearly, this condition is equivalent to ${\bf A}_0$ being invertible, i.e., $\det{\bf A}_0\neq 0$.

One immediate consequence of the determinant condition is worth noting explicitly: a triple $({\bf A}_0, {\bf B}_0 , {\bf v}_0)$ with $2N\times 2N$-matrices ${\bf A}_0, {\bf B}_0$ and ${\bf v}_0\in \R^{2N}$, if drawn randomly from the subset given by \eqref{eqn:gaussian_povm_condition}, according to any nonsingular probability density, defines a minimal informationally complete Gaussian observable with probability one. Hence, almost all Gaussian observables with minimal outcome space are informationally complete.

\section{Covariant Gaussian observables}\label{sec:covariant}
We have seen that a minimal informationally complete Gaussian observable has outcome space $\R^{2N}$, which we can identify with the phase space, whose translations and symplectic transformations act as unitary transformations on the range of the observable. We now look at the consequences of this identification.

An observable $\E:\h B(\R^{2N})\to\h L(\hil^{\otimes N})$ is a covariant phase space observable if it satisfies
$$
W({\bf x}) \E(X) W({\bf x})^* = \E(X+{\bf x}).
$$
A direct computation using the characteristic function \eqref{eqn:gaussian_characteristic_function} and the Weyl relations \eqref{eqn:weyl_relation} shows that a Gaussian phase space  observable $\E$, parametrized by $({\bf A}_0,{\bf B}_0,{\bf v}_0)$, is covariant if and only if
\begin{equation}\label{eqn:gaussian_covariance}
{\bf A}_0=-{\bf \Omega}_N.
\end{equation}
According to Cor.~\ref{cor:IC_gaussian_povm}, every covariant Gaussian phase space observable is thus informationally complete. This result was previously obtained, e.g., in \cite{Ali1977, Werner1984}. As the next proposition shows, any other minimal informationally complete Gaussian observable is connected to a covariant one via linear postprocessing.

\begin{proposition}\label{prop:infocomplete_covariant} A Gaussian observable $\E:\h B(\R^{2N})\to\h L(\hil^{\otimes N})$ is informationally complete if and only if it is a bijective linear postprocessing of a covariant Gaussian observable.
\end{proposition}
\begin{proof} Let $\E^{\rm cov}:\h B(\R^{2N})\to\h L(\hil^{\otimes N})$ be a covariant Gaussian observable, parametrized by $(-{\bf \Omega}_N, {\bf B}^{\rm cov}_0 , {\bf v}^{\rm cov}_0)$, and ${\bf P}$ any invertible matrix. The postprocessing $\E^{\rm cov}_{\bf P}$ is then parametrized by $(-{\bf \Omega}_N{\bf P}^T,{\bf P}{\bf B}^{\rm cov}_0{\bf P}^T,{\bf P}{\bf v}^{\rm cov}_0)$, and is thus informationally complete by Corollary \ref{cor:IC_gaussian_povm}. Conversely, suppose that $\E$ is an informationally complete Gaussian observable, parametrized by $({\bf A}_0, {\bf B}_0 , {\bf v}_0)$. Then ${\bf A}_0$ is invertible, so we can define an invertible matrix ${\bf P}=-{\bf A}_0^T{\bf \Omega}_N$, and further define ${\bf B}_0^{\rm cov}={\bf P}^{-1}{\bf B}_0[{\bf P}^{-1}]^T$, ${\bf v}^{\rm cov}_0={\bf P}^{-1}{\bf v}_0$. Then the triple $(-{\bf \Omega}_N, {\bf B}_0^{\rm cov}, {\bf v}_0^{\rm cov})$ determines a covariant Gaussian observable $\E^{\rm cov}$ for which $\E^{\rm cov}_{\bf P}=\E$.
\end{proof}

Next we take into account the symplectic structure of the phase space. The change \eqref{eqn:coordinate_transformation} of the canonical coordinates via a symplectic matrix ${\bf S}$ transforms a phase space observable $\E:\h B(\R^{2N})\to\h L(\hil^{\otimes N})$ into the phase space observable ${\bf S}(\E)$ given by
\begin{equation}\label{eqn:canonical_transformation_observable}
{\bf S}(\E)(X)= U({\bf S})\E({\bf S}^{-1}(X))U({\bf S}).
\end{equation}
In particular, a covariant Gaussian observable parametrized by $(-{\bf \Omega}_N, {\bf B}^{\rm cov}_0 , {\bf v}^{\rm cov}_0)$ transforms into the covariant Gaussian observable given by $(-{\bf \Omega}_N, {\bf S}{\bf B}^{\rm cov}_0{\bf S}^T , {\bf S} {\bf v}^{\rm cov}_0)$, because ${\bf S} {\bf \Omega}_N{\bf S}^T={\bf \Omega}_N$. Since the positivity condition \eqref{eqn:gaussian_povm_condition} reduces to ${\bf B}^{\rm cov}_0 - i {\bf \Omega}_N \geq {\bf 0}$, we can use Williamson's theorem \cite{Williamson1936} to choose the symplectic matrix ${\bf S}$ such that it diagonalizes ${\bf B}^{\rm cov}_0$, i.e., ${\bf S}{\bf B}^{\rm cov}_0 {\bf S}^T  = \bigoplus_{k=1}^N \beta_k {\bf I}$, where the symplectic eigenvalues of ${\bf B}^{\rm cov}_0$ satisfy $\beta_k\geq 1$. Letting $\mu:\h B(\R^{2N})\to[0,1]$ be the Gaussian probability measure such that $\widehat{\mu}({\bf p}) = e^{-\frac{1}{4} {\bf p}^T \left(\bigoplus_{k=1}^N (\beta_k-1) {\bf I} \right) {\bf p} -i({\bf S v}^{\rm cov}_0)^T{\bf p}}$, the transformed covariant observable is thus
\begin{equation*}
{\bf S}(\E^{\rm cov})=\mu*\E^{Q},
\end{equation*}
where $\E^{Q}$ is the is the $Q$-function of the state, i.e., the covariant Gaussian observable with $(-{\bf \Omega}_N, {\bf I}, {\bf 0})$. To summarize, each covariant Gaussian observable is a Gaussian smearing of the $Q$-function, up to the choice of the canonical coordinates of the phase space. Combining this with Prop.~\ref{prop:infocomplete_covariant}, we have the following result:
\begin{proposition}\label{prop:minimal_infocomplete} A Gaussian observable $\E:\h B(\R^{2N})\to\h L(\hil^{\otimes N})$ is informationally complete if and only if there exist an invertible matrix ${\bf P}$, a symplectic matrix ${\bf S}$, and a Gaussian probability measure $\mu$, such that
\begin{equation*}
\E={\bf S}(\mu*\E^{Q})_{\bf P}.
\end{equation*}
\end{proposition}
In other words, any minimal informationally complete Gaussian observable coincides with a postprocessing of $\E^{\rm Q}$, up to the choice of the canonical coordinates.

Notice that the any covariant Gaussian observable can be measured with the following setup: the signal $N$-mode field is coupled to $N$ auxiliary parameter modes in a Gaussian state   by means of 50:50 beam splitters, and then rotations of $-\pi/2$ are performed on half of the output modes (the case of a single mode is illustrated in Fig.~\ref{fig:eightport}). By varying the Gaussian state of the parameter field, one can obtain an arbitrary covariant Gaussian observable. In particular, a measurement of the $Q$-function is obtained by choosing the parameter field to be the vacuum (i.e., {\bf V}={\bf I}, {\bf m}={\bf 0}). 
\begin{figure}
\includegraphics[width=8cm]{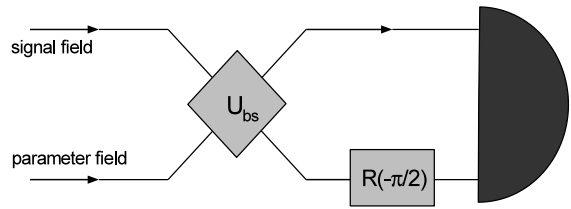}
\caption{Measurement scheme for a single mode covariant Gaussian observable. The signal field is coupled to a parameter field by means of a 50:50 beam splitter, after which a rotation of $-\pi/2$ is performed on one of the output modes. Homodyne detection is then performed on the two output modes.}\label{fig:eightport}
\end{figure}

\section{Commutative and sharp Gaussian observables}
Recall that an observable $\E:\h B(\R^M)\to\h L(\hil^{\otimes N})$ is commutative, if $\E(X)\E(Y)=\E(Y)\E(X)$ for all $X,Y$, and sharp if $\E(X)\E(Y)=\E(X\cap Y)$. It is well-known that an observable is sharp if and only if it is  projection valued, i.e., a spectral measure. 

Clearly, sharp observables are always commutative. Conversely, every commutative observable is obtained from some sharp observable by postprocessing with a Markov kernel (see \cite[Thm.\ 4.4]{Jencova2008}). Smearing with a convolution kernel is a special case, which however turns out to be sufficient to get all commutative Gaussian observables from sharp ones.

In terms of characteristic functions, $\E$ is sharp if and only if the map ${\bf p}\mapsto \widehat{\E}({\bf p})$ is a unitary representation of $\R^M$, i.e.,
\begin{equation}\label{sharp_characteristic_function}
\widehat{\E}({\bf p})\widehat{\E}({\bf p}')=\widehat{\E}({\bf p}+{\bf p}') \text{ for all } {\bf p},{\bf p}'\in \R^{M}.
\end{equation}
Similarly, $\E$ is commutative if and only if
\begin{equation}\label{commutative_characteristic_function}
\widehat{\E}({\bf p})\widehat{\E}({\bf p}')=\widehat{\E}({\bf p}')\widehat{\E}({\bf p}) \text{ for all } {\bf p},{\bf p}'\in \R^{M}.
\end{equation}
This follows by standard approximation arguments from the fact that $\int \widehat{f}({\bf x}) \, \E(\de{\bf x})=\int f({\bf p}) \widehat{\E}({\bf p})\,\de{\bf p}$ holds for any integrable $f:\R^{M}\to \C$, where $\widehat{f}$ is the (suitably normalized) Fourier transform.

In the case of Gaussian observables, the relationship between commutativity and sharpness can now be analyzed explicitly. By \eqref{commutative_characteristic_function}, $\E$ is a commutative if and only if all the Weyl operators $W({\bf A}_0{\bf p})$ commute, which happens exactly when
\begin{equation}\label{eqn:weyl_unitarity_condition}
{\bf A}_0^T {\bf \Omega}_N{\bf A}_0={\bf 0}.
\end{equation}
Note that this condition is exactly the one that makes ${\bf p}\mapsto W({\bf A}_0{\bf p})$ a unitary representation, and thus a characteristic function of a sharp Gaussian observable which we denote by $\sfq_{{\bf A}_0}:\h B(\R^M)\to \h L(\hil^{\otimes N})$. Now if \eqref{eqn:weyl_unitarity_condition} holds, then ${\bf B}_0\geq {\bf 0}$ by the positivity condition \eqref{eqn:gaussian_povm_condition}. Hence, there exists a Gaussian probability measure $\mu:\h B(\R^M)\to[0,1]$, such that $\widehat{\mu}({\bf p}) = e^{-\frac{1}{4}{\bf p}^T{\bf B}_0 {\bf p}-i{\bf v}_0^T{\bf p}}$, and we have
\begin{equation}\label{eqn:commutative_vs_sharp}
\E=\mu*\sfq_{{\bf A}_0}.
\end{equation}
In conclusion, we have proved that a Gaussian observable is commutative exactly when it is a Gaussian smearing of a sharp Gaussian observable.

Since sharp observables are commutative, we can use the \eqref{eqn:commutative_vs_sharp} along with condition \eqref{sharp_characteristic_function} to characterize sharp Gaussian observables. Indeed, since now $\widehat{\E}({\bf p})=\widehat{\mu}({\bf p}) \,\widehat{\sfq}_{{\bf A}_0}({\bf p})$ and the operators $\widehat{\sfq}_{{\bf A}_0}({\bf p})$ are unitary, we have that $\E$ is sharp if and only if $\widehat{\mu}({\bf p+p'})= \widehat{\mu}({\bf p})\widehat{\mu}({\bf p'})$ for all ${\bf p,p'}\in\R^M$. But since $\mu$ is Gaussian, this reduces to the condition ${\bf B_0}+{\bf B}_0^T={\bf 0}$. This implies $({\bf B}_0)_{jj} = -({\bf B}_0)_{jj}=0$, and since ${\bf B}_0\geq {\bf 0}$ by \eqref{eqn:gaussian_povm_condition}, we have ${\bf B}_0={\bf 0}$. Hence, sharp Gaussian observables are parametrized by $({\bf A}_0,{\bf 0},{\bf v}_0)$ with \eqref{eqn:weyl_unitarity_condition}. Note that this can be obtained from the unbiased sharp Gaussian observable $\sfq_{{\bf A}_0}$ by shifting the outcomes with the vector ${\bf v}_0$. Here "unbiased" refers to the fact that for any state with zero expectation for all canonical quadratures $Q_j$, $P_j$, also the expectation of $\sfq_{{\bf A}_0}$ vanishes.

It is known that a single commutative observable is never informationally complete \cite[Thm.\ 2.1.2]{Busch1989}. For commutative Gaussian observables with minimal outcome space $(M=2N)$ this follows immediately from Cor.~\ref{cor:IC_gaussian_povm}, because we have $0=\det {\bf A}_0^T {\bf \Omega}_N{\bf A}_0=2N(\det {\bf A}_0)^2$.  For a general commutative Gaussian observable this can still be seen directly by noting that the rank of ${\bf A}_0$ is at most $N$, as can be seen by applying the Frobenius inequality \cite[Eq.~4.3.3(9), p. 61]{Lutkepohl}:
\begin{eqnarray*}
\rank {\bf A}_0 &=& \frac{1}{2} \left(  \rank {\bf A}_0^T {\bf \Omega}_N   +  \rank {\bf \Omega}_N{\bf A}_0 \right) \\
&\leq &  \frac{1}{2} \left( \rank {\bf \Omega}_N + \rank {\bf A}_0^T {\bf \Omega}_N {\bf A}_0 \right) =N.
\end{eqnarray*}
Cor.~\ref{cor:IC_finite_gaussian} then tells that no finite set of commutative Gaussian observables is informationally complete.

It is worth noting that in the case of a single mode ($N=1$), and minimal outcome space ($M=2$) we have 
$$
{\bf A}_0^T {\bf \Omega} {\bf A}_0 = \left( \begin{array}{cc} 0 & \det {\bf A}_0 \\ -\det {\bf A}_0 & 0 \end{array} \right)
$$
so that $\E$ is informationally complete if and only if $\E$ is noncommutative. In other words, informational incompleteness is equivalent to the observable being a smearing of a sharp observable. Note that this is not true for multiple modes. Indeed, using the above result that $\rank {\bf A}_0\leq N$ for commutative observables, we see that if $N>1$, any given $2N\times 2N$-matrix ${\bf A}_0$ with $N<\rank {\bf A}_0<2N$, together with a ${\bf B}_0$ chosen to satisfy \eqref{eqn:gaussian_povm_condition}, determine a Gaussian observable which is informationally incomplete and noncommutative.

Before proceeding to more specific cases, we make one more observation. It follows from \eqref{eqn:commutative_vs_sharp} that a set of commutative Gaussian observables with matrices ${\bf A}^j_0$, $j\in \h I$ is informationally complete if and only if the set $\{ \sfq_{{\bf A}^j_0}\mid j\in \h I\}$ of the corresponding sharp observables is such. In other words, a set of commutative observables can always be reduced to a set of sharp observables, without affecting informational completeness. Accordingly, we concentrate mostly on sharp observables in the concrete examples.

Consider now the special case of a one dimensional outcome space. Then for any Gaussian observable $\E$, the matrix ${\bf A}_0$ is, in fact, a vector in $\R^{2N}$, and we let ${\bf A}_0 =(a_1,\ldots ,a_{2N})^T$. The condition \eqref{eqn:weyl_unitarity_condition} is automatically satisfied, so we conclude that every Gaussian observable $\E:\h B(\R)\to\h L(\hil^{\otimes N})$ is commutative, and thus a Gaussian smearing of the corresponding sharp observable $\sfq_{{\bf A}_0}$. Furthermore, if $\E:\h B(\R)\to\h L(\hil^{\otimes N})$ is a Gaussian observable with matrix ${\bf A}_0$, then $\E$ is a Gaussian smearing of $\sfq_{{\bf a}}$ with ${\bf a}\in \R^{2N}$, exactly when ${\bf A}_0={\bf a}$.

The observable $\sfq_{{\bf A}_0}$ has the characteristic function $\widehat{\sfq}_{{\bf A}_0}(p) = e^{ip(-{\bf A}_0^T {\bf \Omega }_N{\bf R} )}$ so that  $\sfq_{{\bf A}_0}$ is simply the spectral measure of the selfadjoint operator 
$$
-{\bf A}_0^T {\bf \Omega}_N {\bf R} = \sum_{j=1}^N (a_{2j} Q_j - a_{2j-1} P_j). 
$$
These operators are sometimes  referred to as a generalized field quadratures, and they have also previously been considered in the context of quantum tomography \cite{DAriano1996}. Thm.~\ref{prop:IC_set_of_gaussian_povms} now reduces to the following result:

\begin{proposition}\label{prop:one_dim_case}
Let $\E_j:\h B(\R)\to\h L(\hil^{\otimes N})$, $j\in \h I$ be Gaussian observables with corresponding vectors ${\bf A}_0^j\in \R^{2N}$. Then $\{\E_j\mid j\in \h I\}$ is informationally complete if and only if $\{ \epsilon\, {\bf A}_0^j/\| {\bf A}_0^j\|\mid j\in \h I, \epsilon=\pm 1\}$ is dense in the surface of the unit ball of $\R^{2N}$.
\end{proposition}

We now look at two particular examples in the case of a single mode ($N=1$).

By setting $a_1=-\sin\theta$ and $a_2 = \cos\theta$ we obtain the well known rotated quadrature operators $Q_\theta = \cos\theta\, Q + \sin\theta\, P $, the corresponding observables being denoted by $\sfq_\theta$. From Prop.~\ref{prop:one_dim_case} we immediately see that we can restrict our attention to the values $\theta\in[0,\pi)$, and that a set of the form $\{ \sfq_\theta \mid \theta\in\h I\}$ where $\h I\subset [0,\pi)$ is informationally complete if and only if $\h I$ is dense in $[0,\pi)$. This therefore sharpens the result of, e.g., \cite{Kiukas2008} by proving that density is indeed also necessary. Explicitly, the line on which the Weyl transform can be determined from the measurement of a quadrature $\sfq_\theta$, is given by
$$
X_{\sfq_\theta} = \{ (-p\sin\theta, p\cos\theta)^T \mid p\in\R  \}.
$$

As a slight modification, we set $a_1=-e^{-r}\sin\theta$ and $a_2 = e^r\cos\theta$ so that we obtain the squeezed rotated quadratures $Q_{r,\theta} = e^{r}\cos\theta \,Q + e^{-r}\sin\theta\, P $. In this case we have 
$$
X_{\sfq_{\theta, r}} = \{ (-pe^{-r}\sin\theta, pe^r\cos\theta)^T \mid p\in\R \}
$$
so that for $\theta \neq  k\frac{\pi}{2}$ any change in the value of the squeezing parameter $r$ causes a change in the slope of the line.  For instance, we can fix only two values $\theta_\pm  = \pm\frac{\pi}{4}$ and consider a set $\h I\subset \R$ of squeezing parameters. By Prop.~\ref{prop:one_dim_case}, the set $\{ \sfq_{\theta_\pm, r} \mid r\in\h I\}$ is informationally complete if and only if $$\{(\epsilon_1e^{-r},\epsilon_2e^{r})/\sqrt{2\cosh 2r}\mid r\in \h I,\epsilon_1,\epsilon_2=\pm 1\}$$ is dense in the unit circle, which happens exactly when $\h I$ is dense in $\R$. Indeed, for each of the four choices of signs, the map $\R\ni r\mapsto (\epsilon_1e^{-r},\epsilon_2e^{r})/\sqrt{2\cosh 2r}$ bijectively parametrizes the part of the unit circle lying in the interior of the $(\epsilon_1,\epsilon_2)$-quadrant. In a similar manner we can consider any finite number of values for $\theta$. The benefit of adding more rotations is that for a fixed $\theta$ one only has to find a set of  squeezing parameters which is dense in some interval $[r_-(\theta), r_+(\theta)]$ in order to guarantee informational completeness.

We could also have looked at smearings of sharp quadratures. As already mentioned, this does not add any more structure from the point of view of informational completeness. However, smearings often appear more naturally than sharp observables; we close this section with a particular case involving postprocessing.

Consider the linear postprosessing $\E_{\bf P}$ of a phase space observable $\E$ with matrix ${\bf A}_0$, given by a $1\times 2N$ matrix ${\bf P}$. Such a postprocessing is called a (nonnormalized) marginal of $\E$. Now $\E_{\bf P}$ has one-dimensional outcome space, and is a commutative Gaussian smearing of the generalized quadrature given by the vector ${\bf \tilde A}_0={\bf A}_0{\bf P}^T$. Using this equality, together with Cor. \ref{cor:IC_gaussian_povm}, we now get the following characterization of informational completeness of phase space observables in terms of quadratures.

\begin{proposition}\label{prop:info_complete_connection} Let $\E$ be a Gaussian phase space observable. Then the following are equivalent:
\begin{itemize}
\item[(i)] $\E$ is informationally complete;
\item[(ii)] There is a one-to-one correspondence between generalized quadratures $\sfq_{{\bf a}}$, ${\bf a}\in \R^{2N}$, and linear maps ${\bf P}:\R^{2N}\to \R$, such that $\E_{{\bf P}}$ is a Gaussian smearing of $\sfq_{{\bf a}}$.
\end{itemize}
\end{proposition}

\section{Measurements involving general linear bosonic channels}
As we have noted before, the physical motivation for the definition of Gaussian observables comes from their measurement realizations: any Gaussian POVM can be measured by first applying a Gaussian channel to the system, and then performing homodyne detection. Now Gaussian channels are special cases of what are known as linear bosonic channels, that is, channels  $\Phi:\h T(\hil^{\otimes N}) \to\h T(\hil^{\otimes M})$ which  map Weyl operators according to 
\begin{equation*}
\Phi^*(W({\bf x})) = W({\bf Ax}) f({\bf x})
\end{equation*}
where ${\bf A}$ is a $2N\times 2M$ matrix and $f$ is a complex valued function which must again satisfy suitable complete positivity conditions \cite{Demoen1979, Holevo2001}. The corresponding generalization for Gaussian observables is then given by POVMs $\E:\h B(\R^M)\to\h L(\hil^{\otimes N})$ satisfying 
\begin{equation}\label{eqn:bosonic_povm}
\widehat{\E}({\bf p}) = W({\bf A}_0 {\bf p}) f_0({\bf p}), 
\end{equation}
again with some restrictions on ${\bf A}_0$ and $f_0$.

We immediately notice that  the only difference in the treatment of informational completeness when compared to the Gaussian case is that the function $f$ may be zero at some points. Hence, by measuring $\E$ we can always determine the values of $\widehat{\rho}$ on the set
$$
Y_\E  =  \{{\bf A}_0{\bf p} \mid {\bf p}\in\R^M,  f_0({\bf p}) \neq 0\}\subset   X_\E,
$$
where $X_\E$ is the subspace related to any Gaussian observable with the parameter ${\bf A}_0$. In particular, we have $Y_\E\subset\R^{2N}$ but this is typically not a subspace.  The proof of the following characterization of informational completeness is identical to the proof of Thm.~\ref{prop:IC_set_of_gaussian_povms} and is therefore omitted. 
\begin{theorem}
Let $\E_j$, $j\in\h I$, be observables satisfying \eqref{eqn:bosonic_povm} for some set of parameters $({\bf A}^j_0, f^j_0)$. Then $\{ \E_j\mid j\in\h I\}$ is informationally complete if and only if $\bigcup_{j\in\h I} Y_{\E_j}$ is dense in $\R^{2N}$. 
\end{theorem}
Note that in this more general scenario, a finite set of observables may be informationally complete even though no single observable is such. For instance, suppose that $\E$ and $\E'$ are two observables for which  ${\bf A}_0 ={\bf A}'_0$ but $f_0\neq f'_0$. Then it may happen that $Y_{\E} = \R^{2N}\setminus U$ and $Y_{\E'} = \R^{2N}\setminus U'$ for some open sets $U$ and $U'$ for which $U\cap U'=\emptyset$. In that case neither observable is informationally complete but the pair $\{ \E, \E'\}$ is since $Y_{\E} \cup Y_{\E'} = \R^{2N} $.

As an example, consider covariant phase space observables, i.e., POVMs $\E:\h B(\R^{2N})\to\h L(\hil^{\otimes N})$ satisfying
$$
W({\bf x}) \E(X) W({\bf x})^* = \E(X+{\bf x}).
$$ 
Any covariant phase space observable is generated by a unique positive trace one operator $\sigma$, giving the observable the explicit form  \cite{Holevo1979, Werner1984}
$$
\E_\sigma (X)  = \frac{1}{(2\pi)^N} \int_X W({\bf x}) \sigma W({\bf x})^*\, \de{\bf x}.
$$
In particular, if $\sigma$ is a Gaussian state, then $\E_\sigma$ is a covariant Gaussian observable. The characteristic function of $\E_\sigma$ can be calculated using Eq.~\eqref{eqn:weyl_relation} and the fact that 
$$
 \frac{1}{(2\pi)^N} \int \tr{TW({\bf x}) S W({\bf x})^*}\, \de{\bf x} = \tr{T}\tr{S}
$$
for any $T,S\in\h T(\hil^{\otimes N})$ (see, e.g., \cite[Lemma 1]{Werner1984}). As a result we obtain
$$
\widehat{\E}_\sigma ({\bf p})  = W(-{\bf \Omega}_N {\bf p}) \widehat{\sigma} ({\bf \Omega}_N {\bf p})
$$
so that compared to Eq.~\eqref{eqn:bosonic_povm} we have  ${\bf A}_0 = -{\bf \Omega}_N$, in accordance with Eq.~\eqref{eqn:gaussian_covariance}, and $f_0 ({\bf p }) = \widehat{\sigma} ({\bf \Omega}_N {\bf p})$. Using the fact that $\widehat{\sigma}(-{\bf p}) = \overline{\widehat{\sigma}({\bf p})} = 0 $ if and only if $\widehat{\sigma}({\bf p})=0 $, we find that 
$$
Y_{\E_\sigma} = \{ {\bf p}\in\R^{2N} \mid  \widehat{\sigma} ({\bf p}) \neq 0\}.
$$
In particular, 
 the covariant phase space observable $\E_\sigma$ is informationally complete if and only if ${\rm supp}\, \widehat{\sigma}=\R^{2N}$, that is, the Weyl transform of the generating operator $\sigma$ has full support. This result was obtained also in \cite{Kiukas2012}.

Notice that any covariant phase space observable can be measured with the setup depicted in Sec.~\ref{sec:covariant} and Fig.~\ref{fig:eightport} by choosing the state of the parameter field appropriately (for technical details, see \cite{KiLa2008}). This example actually demonstrates how these more general observables corresponding to linear bosonic channels arise. Indeed, if we take the  dilation of a Gaussian channel as presented in Sec.~\ref{sec:dilation} and replace the state of the auxiliary field by an arbitrary one, then we obtain a linear bosonic channel, and thus an observable,  which is typically not Gaussian but still of the form \eqref{eqn:bosonic_povm}.

Another occasion where these observables may arise is any Gaussian measurement which is subject to classical noise. Indeed, if $\E:\h B(\R^M)\to \h L(\hil^{\otimes N})$ is a Gaussian observable parametrized by $({\bf A}_0,{\bf B}_0,{\bf v}_0)$ and $\mu:\h B(\R^M)\to [0,1]$ is  a probability measure, then $\widehat{\mu *\E} = \widehat{\mu}\, \widehat{\E} $ which corresponds to Eq.~\eqref{eqn:bosonic_povm} with $f_0({\bf p}) = \widehat{\mu} ({\bf p}) e^{-\frac{1}{4} {\bf p}^T {\bf B}_0 {\bf p}-i{\bf v}_0^T {\bf p}}$. Hence, 
$$
Y_{\mu*\E} = \{ {\bf A}_0{\bf p} \mid {\bf p}\in\R^M , \widehat{\mu}({\bf p})\neq 0 \}.
$$
In particular, if $\E:\h B(\R^{2N})\to \h L(\hil^{\otimes N})$ is an informationally complete phase space observable, then $\mu*\E$ is informationally complete if and only if ${\rm supp}\, \widehat{\mu} = \R^{2N}$.

\section{Conclusions}
We have proved a characterization for the informational completeness of an arbitrary set of  Gaussian observables. As a consequence, we have shown that unless one has access to a single informationally complete Gaussian observable, then one needs infinitely many observables. We have characterized informationally complete Gaussian observables which are minimal in the sense of having the smallest possible dimension of the outcome space, as the observables which are bijective linear postprocessing of covariant Gaussian phase space observables. We have then developed this further and shown that any minimal informationally complete Gaussian observable is actually a postprocessing of the observable $\E^{\rm Q}$ for which the outcome distribution is the $Q$-function of the state, given a suitable symplectic coordinatization of the phase space. We have also treated commutative Gaussian observables separately, and shown that infinitely many such observables are needed in order to reach informational completeness. As a special case we have characterized informationally complete sets of generalized field quadratures, i.e., sharp Gaussian observables with one-dimensional outcome space, and proved a connection between informationally complete Gaussian phase space observables and generalized field quadratures.

Since Gaussian observables can be measured by combining Gaussian channels with homodyne detection, we have also studied the natural generalization to the case where the Gaussian channel is replaced by a linear bosonic channel. Also in this case we have obtained necessary and sufficient conditions for the informational completeness of any set of such observables.


\section*{Acknowledgments}
JS acknowledges support from the Academy of Finland (grant no. 138135) and the Italian Ministry of Education, University and Research (FIRB project RBFR10COAQ). JK acknowledges support from the European CHIST-ERA/BMBF project CQC, and the \mbox{EPSRC} project EP/J009776/1.

\end{document}